\documentclass[a4paper]{amsart}
\usepackage{microtype}

\usepackage{amssymb}
\usepackage{complexity}
\usepackage{braket}
\usepackage[hidelinks]{hyperref}
\usepackage{orcidlink}

\usepackage{tikz}
\usetikzlibrary{quantikz2}
\usepackage{adjustbox}

\usepackage{thmtools} % to fix \autoref
\newtheorem{theorem}{Theorem}
\newtheorem{lemma}[theorem]{Lemma}
\newtheorem{corollary}[theorem]{Corollary}
\newtheorem{fact}[theorem]{Fact}

\theoremstyle{definition}
\newtheorem{definition}[theorem]{Definition}

\newtheorem{example}[theorem]{Example}

\DeclareMathOperator{\cert}{cert}
\DeclareMathOperator{\pre}{pre}
\DeclareMathOperator{\var}{var}
\DeclareMathOperator{\val}{val}
\DeclareMathOperator{\depth}{depth}

\let\int\relax
\DeclareMathOperator{\int}{int}

\newcommand{\n}{\overline}
\newcommand{\gateset}{\mathcal G}
\renewcommand{\C}{\mathbb C}
\renewcommand{\R}{\mathbb R}

\title{Quantum state preparation for weighted \MakeLowercase{d}-DNNF}
\author[S. Hegeman]{Steef Hegeman\,\orcidlink{0009-0009-0368-8502}}
\address{Leiden University}
\email{s.hegeman@liacs.leidenuniv.nl}
\author[J.\,H. Lee]{Joon Hyung Lee\,\orcidlink{0000-0001-8628-9922}}
\email{j.h.lee@liacs.leidenuniv.nl}
\author[A. Laarman]{Alfons Laarman\,\orcidlink{0000-0002-2433-4174}}
\email{a.w.laarman@liacs.leidenuniv.nl}

\begin{document}

\begin{abstract}
The quantum state preparation problem is to,
given a description of a quantum state, 
efficiently generate a quantum circuit computing the state.
We show that for quantum states described by weighted d-DNNF
(deterministic, decomposable pseudo-Boolean circuits)
a quantum circuit computing the state can be obtained in linear time
up to complex arithmetic.
\end{abstract}

\maketitle

\section{Introduction}\label{sec:intro}
Many well-known quantum algorithms
assume that their input is available as some entangled quantum state.
The quantum state preparation problem (QSP)
of actually preparing this state from a classical description
is hard in general \cite{sun_asymptotically_2023}
and can even cancel out an algorithm's quantum advantage \cite{aaronson2015read}.
We show how QSP can be done efficiently for quantum states represented as weighted d-DNNF circuits.

In particular,
we show that for any weighted d-DNNF
(a deterministic, decomposable pseudo-Boolean circuit) $D : \{ 0,1 \}^n \to \C$,
state $\sum_x D(x) \ket x$ can be prepared by a quantum circuit
of size $O(|D|)$ and depth $O(\depth(D) \log |D|)$,
and that a description of this quantum circuit can be computed classically in $O(|D|)$ time up to complex arithmetic (in a computational RAM model where operations $+, \times, \div, |\cdot|, \sqrt{\cdot}$ can be performed on complex number registers in constant time).

Unweighted d-DNNF \cite{darwiche_decomposable_2001, darwiche_tractable_2001}
is a well-known representation language
from the field of classical knowledge compilation \cite{darwiche_knowledge_2002}
that is exponentially more succinct than the internal language of contemporary \#SAT model counting algorithms \cite{huang_language_2007,beame2013}.
As a quantum state representation language,
weighted d-DNNF is exponentially more succinct than weighted FBDD
(see \autoref{sec:related-work})
as recently prepared in \cite{tanaka2026qsp}.

\subsection*{Structure of the paper}
Section~\ref{sec:defs} defines the concepts used in the paper.
Section~\ref{sec:main-results} states the main results.
Section~\ref{sec:overview} gives an informal overview of the state preparation algorithm.
Sections~\ref{sec:impl}~to~\ref{sec:reweighting} prove the technical results.
In \autoref{sec:related-work} we discuss and compare with related work.

\section{Definitions}\label{sec:defs}
\subsection{Weighted d-DNNF}
Following \cite{quist2026},
a \emph{pseudo-Boolean circuit}
(over $\C$)
is an arithmetic circuit
consisting of $\times$-gates $g = \prod_i g_i$
and weighted $+$-gates $g = \sum_i w_i g_i$ with weights $w_i \in \C$,
whose leaves are of the form $x$ and $\n x$
for Boolean variables $x$.
Let $\var(g)$ refer to the set of variables $x$
so that $g$ depends syntactically on a leaf $x$ or $\n x$.
Given an assignment $\alpha : \var(g) \to \{ 0, 1 \}$,
the evaluation $g(\alpha)$ is obtained by replacing all $x$ with $\alpha(x)$
and $\n x$ with $1 - \alpha(x)$
and evaluating the circuit.
Pseudo-Boolean circuits
can be thought of as Boolean circuits in negation normal form (NNF) that weigh their assignments.
In this spirit,
we will sometimes call $\times$-gates AND-gates and
$+$-gates OR-gates.
A pseudo-Boolean circuit $D$
can describe a quantum state
$\sum_x D(x) \ket x$.

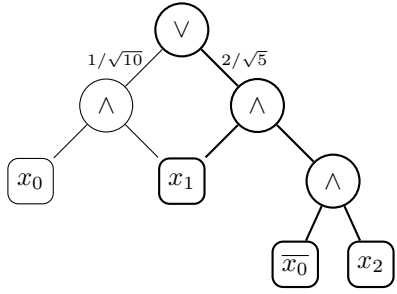
\begin{figure}
\begin{tikzpicture}[
	gt/.style={draw,circle,minimum size=7mm,inner sep=0pt},
	lit/.style={draw,rectangle,rounded corners,minimum size=6mm,inner sep=2pt},
	wt/.style={font=\tiny},
	every label/.style={font=\footnotesize},
	cert/.style={thick},
]
	\node[gt, cert] (g0) at (0,0)    {$\lor$};
	\node[gt,
		%label=left:{$g_1$}
		]  (g1) at (-1,-1) {$\land$};
	\node[gt, cert,
		%label=right:{$g_2$}
		] (g2) at (1,-1)  {$\land$};
	\node[lit] (x0)  at (-2,-2) {$x_0$};
	\node[lit, cert] (x1)  at (0,-2) {$x_1$};
	\node[gt, cert,
		%label=right:{$g_3$}
		] (g3) at (2,-2)  {$\land$};
	\node[lit, cert] (nx0)  at (1.5,-3.1) {$\n{x_0}$};
	\node[lit, cert] (x2)  at (2.5,-3.1) {$x_2$};

	\draw
		(g0) -- node[wt,left,yshift=3,xshift=3] {$1/\sqrt {10}$} (g1)
		(g1) -- (x0)
		(g1) -- (x1)
		;
	\draw[cert]
		(g0) -- node[wt,right,yshift=3,xshift=-3] {$2/\sqrt 5$} (g2)
		(g2) -- (x1)
		(g2) -- (g3)
		(g3) -- (nx0)
		(g3) -- (x2)
		;
\end{tikzpicture}
\caption{
	A weighted 2-d-DNNF $D$.
	It has three nonzero assignments and two certificates.
	The certificate of $D(0,1,1) = 2/\sqrt 5$ is shown in bold.
	$S_1(D) = 1/\sqrt{10} \ket{110} + 1/\sqrt{10}\ket{111} + 2/\sqrt 5 \ket{011}$.
}
\label{fig:example-2-d-DNNF}
\end{figure}

\emph{Weighted d-DNNF} are pseudo-Boolean circuits that are deterministic and decomposable.
\emph{Determinism} requires that for every $+$-gate $g = \sum_i w_i g_i$
the function
$x \mapsto w_i g_i(x) \cdot w_j g_j(x)$
is the zero function for all distinct $i,j$.
(In particular, if $g(x)$ is nonzero then exactly one $w_i g_i(x)$ is nonzero.)
\emph{Decomposability} requires that for every $\times$-gate $\prod_i g_i$,
the $g_i$ do not share variables:
$\var(g_i) \cap \var(g_j) = \varnothing$ for distinct $i,j$.
Without weights, this definition corresponds with constant-free d-DNNF of \cite{darwiche_tractable_2001}.

To improve readability,
we assume without loss of generality
that the output gate is not an input gate
(at least one $+$- or $\times$-gate),
that all weights are nonzero,
that every gate is (indirectly) connected to the output gate,
and that there is at most one positive and one negative literal per variable.
Our definition excludes constant gates,
though any nontrivial weighted d-DNNF with constants
can be made constant-free in linear time
(\autoref{app:constant-free}).

For most of the technical sections,
we work with ``nice'' bounded-degree d-DNNF.

\begin{definition}
A (weighted) \emph{2-d-DNNF}
is a (weighted) d-DNNF where all gates have fan-out and fan-in at most 2,
and there are no unary AND-gates.
\end{definition}

See \autoref{fig:example-2-d-DNNF}.
The assumption that there are no unary AND-gates
is made solely for readability purposes
(they are equivalent to unary OR-gates with weight 1).
Any weighted d-DNNF $D$ can be transformed into an equivalent weighted 2-d-DNNF $D'$
of depth $O(\depth(D)\log|D|)$
in linear time up to complex arithmetic.

We call a weighted d-DNNF $D$
($\ell_2$-)\emph{normalized}
if for all OR-gates with weights $w_i$
it holds that $\sum_i |w_i|^2 = 1$.

\subsection{Certificates}
Our construction rests on a key property of d-DNNF:
uniqueness of certificates.

\begin{definition}
	Let $D$ be a (weighted) d-DNNF.
	A \emph{certificate} of $D$ for assignment $x$
	is a smallest subcircuit of $D$ evaluating
	to a nonzero value under assignment $x$
	that includes the output gate.
	(So a certificate for $x$ witnesses $D(x) \neq 0$.)
\end{definition}

In other words: a certificate for $D(x) \neq 0$
is a connected subcircuit of $D$
that includes the output gate,
contains all children of AND-gates it contains,
exactly one child of every OR-gate it contains,
and is nonzero under the assignment $x$.

The following well-known fact follows from determinism and decomposability,
with the tree part using in addition that our definition of d-DNNF is constant-free.

\begin{fact}[Uniqueness of certificates]\label{fact:tree-certs}
	For every (weighted) d-DNNF,
	every assignment $x$ has at most one certificate,
	and every certificate is a tree.
\end{fact}

That is, seen as a subgraph of the d-DNNF,
a certificate is a tree:
every gate in a certificate feeds into at most one other gate in the certificate.

Given a d-DNNF $D$,
we write $\cert(D)$ for the set of certificates of $D$.
We write $T_D(x)$ for the certificate of $x$ in $D$ if it exists ($D(x) \neq 0$),
or $\varnothing$ otherwise ($D(x) = 0$).
If the context is clear,
we will drop the $D$ and write $T(x)$.
From determinism it can be shown that if $D(x) \neq 0$
then $D(x) = \val(T_D(x))$.
For weighted d-DNNF $D$, if $D(x) \neq 0$ then $D(x)$ is equal to the product of the edge weights in $T_D(x)$.

For any circuit $c$,
write $\var(c)$ for the set of variables mentioned in $c$
and $\int(c)$ for the set of non-input gates.
If it can be evaluated, write $\val(c)$ for its value.
If $g$ is a gate in some circuit $D$,
we will sometimes refer with $g$ to the maximal subcircuit of $D$ rooted at $g$,
so that $\cert(g), \var(g), \val(g)$ are defined.
(In particular, if $r$ is the output gate of $D$,
then $r$ refers to the circuit $D$.)

While every assignment of a d-DNNF $D$ has at most one certificate,
multiple assignments can have the same certificate.
A certificate $\tau$ certifies
$m_D(\tau) = 2^{|\var(D)| - |\var(\tau)|}$ assignments.
For assignments $x$,
we write $m_D(x)$ for $m_D(T(x))$.

\subsection{Quantum state preparation}\label{subsec:qsp}
Given a weighted d-DNNF $D$,
we define two associated vectors:
\begin{align*}
	S_1(D) &= \sum_{x \in 2^{|\var(D)|}} D(x) \ket x\\
	S_2(D) &= \sum_{x \in 2^{|\var(D)|}} \frac{D(x)}{\sqrt{m_D(x)}}\ket x
	= \sum_{\tau \in \cert(D)} \val(\tau) \chi(\tau)
\end{align*}
where $\chi(\tau)_x$ is $\ket 1, \ket 0$
if $x$ appears positively, negatively in $\tau$,
and $\ket +$ otherwise.

If $D$ is smooth \cite{darwiche_tractable_2001},
meaning that $\var(\tau) = \var(D)$ for all certificates $\tau$,
then $S_1(D) = S_2(D)$.
While the usual objective of quantum state preparation would be
a circuit computing $S_1(D)$,
the second is more natural from a certificate standpoint.

The quantum circuits we produce will
use gates only from $\{ X, CX \} \cup SU(2)$,
that is:
two-qubit gates $CX$
and single-qubit gates $X$ and $U(p,q)$
given by
\begin{align*}
	U(p,q) =
	\begin{pmatrix}
		p & -q^*\\
		q & p^*
	\end{pmatrix}
\end{align*}
for $p,q$ with $|p|^2 + |q|^2 = 1$.
For readability,
we will initially assume a gate set $\gateset$
that additionally includes $CCX$, $H$, $CH$, $CCH$, and $CU(p,q)$.
In \autoref{app:gate-set}
we show how to reduce this
to using just $X, CX$ and $U(p,q)$.

\subsection{Gate notation}
We write $C_{a_0}\dots C_{a_n}G_b$
for the quantum gate $G$ applied to qubit $b$
controlled on qubits $a_0, \dots, a_n$.
Negated controls are indicated with $\n a$,
so that $C_{\n a}G_b$ is given by $X_a \circ C_a G_b \circ X_a$.

\subsection{Circuit dimensions}
For classical circuits $D$,
the \emph{depth} of a gate $g$
is given by the length of the longest path from the output gate of $D$ to $g$.
The depth of $D$ is equal to that of its deepest gate.
The $i$th \emph{level} of a circuit is the set of its gates that have depth $i$.
The \emph{width} of a level is the number of gates in it,
and the width of a circuit is the width of its widest level.
We define the size $|D|$ of a circuit to be the size of its underlying graph:
the number of gates plus the number of wires.

For quantum circuits,
width is equal to the number of qubits the circuit targets.
Because two gates cannot act on the same qubit in parallel,
circuit depth depends on how conflicting gates are scheduled.
With $\depth(Q)$ we refer to the smallest possible depth
(of an optimal scheduling).

\section{Main results}\label{sec:main-results}
Our main technical result
(proved in sections~\ref{sec:impl}~to~\ref{sec:complexity})
is the following.

\begin{theorem}\label{thm:s2}
	For every normalized weighted 2-d-DNNF $D$ with $S_2(D)$ a state,
	there is a quantum circuit computing it
	with $O(|D|)$ gates,
	depth $O(\depth(D))$,
	and $2|\int(D)| - 2$ ancillas.
	Up to complex arithmetic,
	it can be described in linear time
	and in constant parallel time with $\poly(|D|)$ processors.
\end{theorem}

In \autoref{sec:reweighting},
we will show that weighted 2-d-DNNF $D$
where $S_2(D)$ is a state
can be normalized
while preserving the represented state,
and that weighted 2-d-DNNF $D$
where $S_1(D)$ is a state can be reweighted to obtain normalized weighted 2-d-DNNF $D'$ with $S_1(D) = S_2(D')$.
Both operations only modify edge weights
and can be performed in linear time up to complex arithmetic.
It then follows:

\begin{theorem}\label{thm:2-d-DNNF}
	For every weighted 2-d-DNNF $D$
	where $S_1(D)$ or $S_2(D)$ is a state,
	there is a quantum circuit computing it
	with $O(|D|)$ gates,
	depth $O(\depth(D))$,
	and $2|\int(D)| - 2$ ancillas.
	Up to complex arithmetic,
	it can be described in linear time.
\end{theorem}

There are many ways to convert arbitrary d-DNNF to 2-d-DNNF
with different (asymptotic) size/depth/width trade-offs.
One choice gives:

\begin{corollary}\label{cor:main}
	For every weighted d-DNNF $D$ with $S_1(D)$ or $S_2(D)$ a state,
	there is a quantum circuit computing it
	with $O(|D|)$ gates,
	width $O(|D|)$,
	and depth $O(\depth(D)\log |D|)$.
	Up to complex arithmetic,
	it can be described in linear time.
\end{corollary}

In \autoref{app:gate-set},
it is proved that all circuits mentioned above can be made to only use gates from the gate set
$\{ X, CX \} \cup SU(2)$.

\section{Construction overview}\label{sec:overview}
In \autoref{sec:reweighting},
we show that preparing $S_1(D')$ for weighted 2-d-DNNF $D'$
reduces efficiently to preparing states $S_2(D)$ for normalized weighted 2-d-DNNF $D$.
Here we give an overview of the quantum circuit $Q(D)$ preparing $S_2(D)$ for such normalized weighted 2-d-DNNF $D$.

There are qubits $c_g$ for every internal gate and $n_g$ for every internal gate except the output gate,
as well as qubits $v_x$ for every variable $x$.
Initially, $c_r = \ket 1$,
where $r$ is the root (output) gate,
$v_x = \ket +$ for all $x$,
and $c_g = \ket 0$ for all $g \neq r$.

The circuit roughly consists of three layers:
the certificate propagation layer,
the nonzero layer,
and the reset layer.

After the certificate propagation layer,
the $c_g$ and $v_x$ wires will be in superposition over all certificates,
weighted according to $D$.
Initially, with $c_r = \ket 1$ and the $v_x$ set to $\ket +$,
it is a superposition over all partial certificates of depth 1.
It will be expanded to partial certificates of higher depth as follows:
\begin{itemize}
	\item
		For every AND-gate $g = a \land b$
		there are quantum gates $C_{c_g}X_{c_a}, C_{c_g}X_{c_b}$,
		setting $c_a,c_b$ to $\ket 1$ if $c_g$ is.
		This corresponds with the rule that if an AND-gate is in a certificate,
		then all of its children are also in the certificate.
	\item
		For OR-gates $g = a \lor b$ with weights $p,q$,
		there are gates $C_{c_g}U_{c_a}(q,p)$
		and $C_{c_g}C_{\n{c_a}}X_{c_b}$,
		setting $c_a c_b$ to $\ket {10}$ with probability $|p|^2$
		and to $\ket {01}$ otherwise.
		This corresponds with the rule that if an OR-gate is in a certificate,
		then exactly one of its children is also in the certificate.
		So the partial certificate ``splits'' here into two partial certificates of higher depth,
		with their relative weight determined by the weights on the OR-gate.
\end{itemize}
This continues all the way to the $v_x$-wires.
Here, extra care has to be taken
for the initial $\ket +$
and negative literals.
Now the $c_g,v_x$ qubits together give the proper distribution over the certificates
and corresponding assignments.

It remains to reset the $c_g$ to $\ket 0$.
In preparation,
in the nonzero layer,
the circuit first computes for each assignment $\alpha$ appearing in the $v_x$
the gates $g$ in $D$ that are nonzero under that assignment,
setting $n_g = \ket 1$ if $g(\alpha) \neq 0$ and $n_g = \ket 0$ otherwise.
This is achieved by quantum gates $C_{n_a}C_{n_b}X_{n_g}$
for $g = a \land b$
and $C_{n_a}X_{n_g}, C_{n_b}X_{n_g}$ for $g = a \lor b$,
where the $n_a,n_b$ are replaced with (negated) $v_x,v_y$
in case the gate takes literal inputs.

In the reset layer,
with the $n_g$ in hand,
the $c_a$ can be reset as follows,
from bottom to top:
for all $g = \land_i a_i$
there are $C_{c_g}X_{c_{a_i}}$,
corresponding with the rule that if an AND-gate is in a certificate then all its children are,
and for all $g = \lor_i a_i$
there are $C_{c_g}C_{n_{a_i}}X_{c_{a_i}}$,
corresponding with the rule that if an OR-gate is in a certificate then exactly one of its children is,
namely the one that is nonzero under any of the certificate's assignments.

Finally,
the $n_g$ are reset by inverting their computation
(since apart from their computation,
they are only used as controls).

\begin{figure}
\begin{tikzpicture}[
	gt/.style={draw,circle,minimum size=7mm,inner sep=0pt},
	lit/.style={draw,rectangle,rounded corners,minimum size=6mm,inner sep=2pt},
	wt/.style={font=\tiny},
	every label/.style={font=\footnotesize}
]
	\node[gt,label=above:{$r=g_0$}] (g0) at (0,0)    {$\lor$};
	\node[gt,label=left:{$g_1$}]  (g1) at (-1,-1) {$\land$};
	\node[gt,label=right:{$g_2$}] (g2) at (1,-1)  {$\land$};
	\node[lit] (x0)  at (-2,-2) {$x_0$};
	\node[lit] (x1)  at (0,-2) {$x_1$};
	\node[gt,label=right:{$g_3$}] (g3) at (2,-2)  {$\land$};
	\node[lit] (nx0)  at (1.5,-3.1) {$\n{x_0}$};
	\node[lit] (x2)  at (2.5,-3.1) {$x_2$};

	\draw
		(g0) -- node[wt,left,yshift=3,xshift=3] {$1/\sqrt 5$} (g1)
		(g0) -- node[wt,right,yshift=3,xshift=-3] {$2/\sqrt 5$} (g2)
		(g1) -- (x0)
		(g1) -- (x1)
		(g2) -- (x1)
		(g2) -- (g3)
		(g3) -- (nx0)
		(g3) -- (x2);
\end{tikzpicture}
\caption{
	A normalized weighted 2-d-DNNF $D$
	with two certificates.
	$S_2(D) = 1/\sqrt{5} \ket{11+} + 2/\sqrt 5 \ket{011}$.
	(The $S_2$-state of $D$ is the $S_1$-state of the weighted d-DNNF in \autoref{fig:example-2-d-DNNF}.)
}
\label{fig:example-D}
\end{figure}
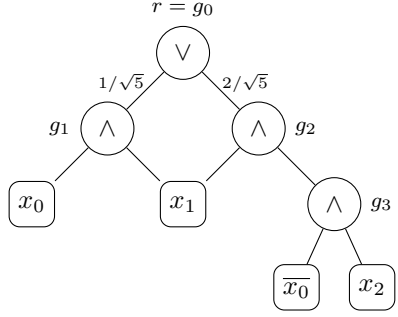

\begin{example}
\begin{figure}
\centering
\begin{adjustbox}{max width=\textwidth}
\begin{quantikz}[row sep={20pt, between origins},column sep={4pt}]
	\lstick{$c_{g_0}=\ket1$} & \ctrl{1} & \ctrl{2} & \qw & \qw & \qw & \qw & \qw & \qw & \qw & \qw & \qw & \qw\slice[style=gray]{} & \qw & \qw & \qw\slice[style=gray]{} & \qw & \ctrl{1} & \ctrl{2} & \gate{X} & \qw & \qw & \qw & \qw \\
	\lstick{$c_{g_1}=\ket0$} & \gate{U(\frac{2}{\sqrt5},\frac{1}{\sqrt5})} & \octrl{1} & \ctrl{6} & \ctrl{6} & \ctrl{7} & \ctrl{7} & \qw & \qw & \qw & \qw & \qw & \qw & \qw & \qw & \qw & \qw & \targ{} & \qw & \qw & \qw & \qw & \qw & \qw \\
	\lstick{$c_{g_2}=\ket0$} & \qw & \targ{} & \qw & \qw & \qw & \qw & \ctrl{6} & \ctrl{6} & \ctrl{1} & \qw & \qw & \qw & \qw & \qw & \qw & \ctrl{1} & \qw & \targ{} & \qw & \qw & \qw & \qw & \qw \\
	\lstick{$c_{g_3}=\ket0$} & \qw & \qw & \qw & \qw & \qw & \qw & \qw & \qw & \targ{} & \ctrl{4} & \ctrl{6} & \ctrl{6} & \qw & \qw & \qw & \targ{} & \qw & \qw & \qw & \qw & \qw & \qw & \qw \\
	\lstick{$n_{g_1}=\ket0$} & \qw & \qw & \qw & \qw & \qw & \qw & \qw & \qw & \qw & \qw & \qw & \qw & \targ{} & \qw & \qw & \qw & \ctrl{-3} & \qw & \qw & \qw & \qw & \targ{} & \qw \\
	\lstick{$n_{g_2}=\ket0$} & \qw & \qw & \qw & \qw & \qw & \qw & \qw & \qw & \qw & \qw & \qw & \qw & \qw & \qw & \targ{} & \qw & \qw & \ctrl{-3} & \qw & \targ{} & \qw & \qw & \qw \\
	\lstick{$n_{g_3}=\ket0$} & \qw & \qw & \qw & \qw & \qw & \qw & \qw & \qw & \qw & \qw & \qw & \qw & \qw & \targ{} & \ctrl{-1} & \qw & \qw & \qw & \qw & \ctrl{-1} & \targ{} & \qw & \qw \\
	\lstick{$v_{x_0}=\ket+$} & \qw & \qw & \gate{H} & \targ{} & \qw & \qw & \qw & \qw & \qw & \gate{H} & \qw & \qw & \ctrl{-3} & \octrl{-1} & \qw & \qw & \qw & \qw & \qw & \qw & \octrl{-1} & \ctrl{-3} & \qw \\
	\lstick{$v_{x_1}=\ket+$} & \qw & \qw & \qw & \qw & \gate{H} & \targ{} & \gate{H} & \targ{} & \qw & \qw & \qw & \qw & \ctrl{-4} & \qw & \ctrl{-3} & \qw & \qw & \qw & \qw & \ctrl{-3} & \qw & \ctrl{-4} & \qw \\
	\lstick{$v_{x_2}=\ket+$} & \qw & \qw & \qw & \qw & \qw & \qw & \qw & \qw & \qw & \qw & \gate{H} & \targ{} & \qw & \ctrl{-3} & \qw & \qw & \qw & \qw & \qw & \qw & \ctrl{-3} & \qw & \qw
\end{quantikz}
\end{adjustbox}
\caption{
	The quantum circuit $Q(D)$ preparing $S_2(D)$
	for the 2-d-DNNF $D$ from \autoref{fig:example-D}.
	The dashed lines show the boundaries between
	the certificate propagation,
	nonzero,
	and reset layers.
}
\label{fig:example-Q(D)}
\end{figure}
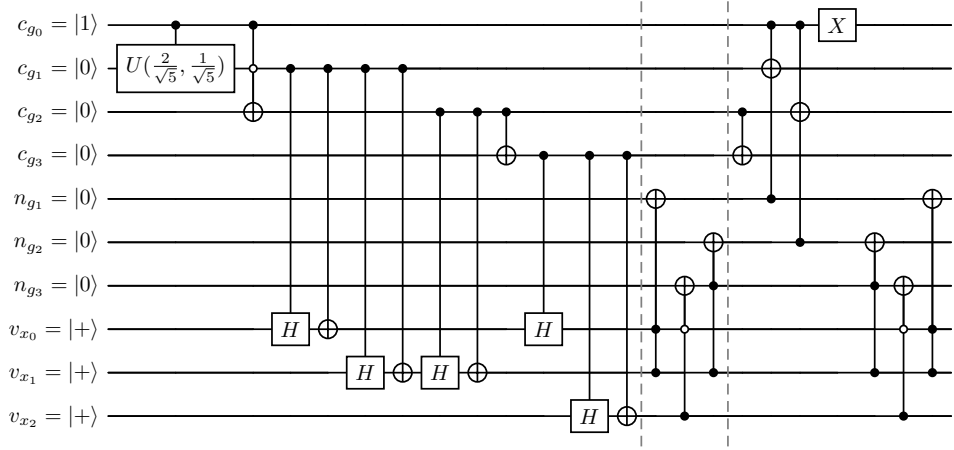

Consider the normalized 2-d-DNNF $D$ from \autoref{fig:example-D}.
The quantum circuit $Q(D)$ preparing $S_2(D)$ is displayed in \autoref{fig:example-Q(D)}.
It first applies the $P_{g_i}$ in the topological order $g_0,g_1,g_2,g_3$
(their subgates are the ones controlled on $c_{g_i}$).
The first dashed line indicates the start of the nonzero layer.
In the reset layer, starting after the second dashed line,
first the $R_g$ are applied.
Then the $n_g$ are reset
simply by inverting the nonzero layer.
It is important to note that,
even though there are multiple $CH$ gates applied to $v_{x_1}$,
in every basis state of nonzero amplitude at most one is fired,
since every (partial) certificate contains each variable at most once.
\end{example}

\section{The state preparation algorithm}\label{sec:impl}
\newcommand{\I}{\mathcal I}
Let $D = (G, E)$ be a normalized
weighted 2-d-DNNF so that $S_2(D)$ is a quantum state.
We describe how to construct a quantum circuit $Q(D)$ preparing
(with some ancillas) state $S_2(D)$.

Let $r \in G$ be the output gate of $D$.
Write $\I$ for the internal gates $\int(D)$.

There is an ancilla qubit $c_g$ for every $g \in \I$,
an additional ancilla $n_g$ for every $g \in \I \setminus \{ r \}$,
and finally there are qubits $v_x$ for every $x \in \var(D)$.
Initially, set $c_r = \ket 1$,
$v_x = \ket +$ for all $x \in\var(D)$,
and $c_g = \ket 0$ for all other $g$.
Set all $n_g$ to $\ket 0$.

The circuit $Q(D)$ consists of three layers:
\begin{enumerate}
	\item The certificate layer,
		where $D$ is uncomputed top-to-bottom to find the certificates,
		storing them on the $c_g$ wires
	\item The nonzero layer,
		where $D$ is computed bottom-to-top to find nonzero gates
		under given assignments,
		storing them on the $n_g$ wires
	\item The reset layer,
		where the $n_g$ are used to reset the internal $c_g$,
		and then the $n_g$ are reset by inverting layer (2).
\end{enumerate}

We now describe how the layers are implemented.
\begin{enumerate}
\item 
	Layer (1) consists of subcircuits $P_g$ for $g \in \I$.
	Let $F_a$ be $H_{v_x}$ if $a$ is $\n x$;
	$(X \circ H)_{v_x}$ if $a$ is $x$;
	and $X_{c_a}$ otherwise.
	The $P_g$ are as follows.
	\begin{itemize}
		\item If $g$ is an AND-gate $a \land b$
			then $P_g$ consists of
			$C_{c_g}F_a$ and $C_{c_g}F_b$.
		\item If $g$ is an OR-gate $\lor_{i < m} a_i$
			with weights $w_i$
			and $m \leq 2$,
			then $P_g$ is as follows.
			If $m = 1$ then let $w_1 = 0$.
			Assume without loss of generality
			that the $a_i$ are sorted so that internal gates come before literals
			and positive literals come before negative ones.

			If $g$ is $x \lor \n x$
			then $P_g$ consists of $C_{c_g}H_{v_x}$
			followed by $C_{c_g}U_{v_x}(w_1,w_0)$.

			Otherwise do the following.
			If $a_0$ is a literal for variable $x$,
			then first apply $C_{c_g}H_{v_x}$.
			% set to 0
			Apply $C_{c_g}U_{a_0}(w_1,w_0)$,
			where $a_0$ refers to $v_x$ if $a_0$ is a literal for $x$
			and $c_{a_0}$ otherwise.
			% now we have p|1> + q|0> in a
			Apply $C_{c_g}C_{\n {a_0}}F_{a_1}$ if $m = 2$.
			% now b is p|0> + q|1>
			Finally, if $a_0$ is $\n x$,
			apply $C_{c_g}X_{v_x}$.
	\end{itemize}

\item
	Layer (2) computes $D(x) \neq 0$ into the $n_g$,
	where $x$ is given by the $v_x$.
	That is, it consists of subcircuits $N_g$
	where $N_g$ is
	$C_{n_a}X_{n_g}$ if $g$ is a unary OR-gate with input $a$,
	$C_{n_a}C_{n_b}X_{n_g}$ if $g$ is of the form $a \land b$,
	and $C_{n_a}X_{n_g}$ followed by $C_{n_b}X_{n_g}$ if $g$ is $a \lor b$,
	where the $n_a,n_b$ are substituted with $v_x$ where $a,b$ is a variable $x$,
	and negated control on $v_x$ where $a,b$ is $\n x$.
\item
	Layer (3) consists of reset subcircuits $R_g$
	where the $n_a$ are only used as controls,
	followed by an $X$-gate on $c_r$,
	and finally by an inversion of layer (2) to reset the $n_a$.
	The $R_g$ are as follows.
	\begin{itemize}
		\item For unary OR-gates and AND-gates $g$
			with children $a_i$,
			$R_g$ consists of all
			$C_{c_g}X_{c_{a_i}}$ with $a_i \in \I$.
		\item For $g = a_0 \lor a_1$,
			$R_g$ consists of all $C_{c_g}C_{n_{a_i}}X_{c_{a_i}}$
			with $a_i \in \I$.
	\end{itemize}
\end{enumerate}

The $P_g$ are applied from top to bottom
and the $N_g$ and $R_g$ from bottom to top.
(More precisely:
$P_g$ is scheduled before all $P_{g'}$ for gates $g'$ that $g$ depends on,
and the $N_g,R_g$ respectively
are scheduled after $N_{g'},R_{g'}$ for gates $g'$ that $g$ depends on.
Gates acting on disjoint qubits commute,
so independent gates are applied in parallel,
effectively following the structure of $D$.)
Write $Q(D)$ for the quantum circuit obtained from $D$ by the description above.

\section{Correctness}\label{sec:correctness}
We now prove that $Q(D)$ enters $S_2(D)$ as desired.
Fix some topological order $r = g_0 \prec \dots \prec g_{|G| - 1}$ on the gates $G$ of $D$.

A \emph{precertificate} $\rho$ of $D$
is a subcircuit of $D$ obtained by starting with $\rho = \{ r \}$
and then iteratively adding both children of an AND-gate in $\rho$ to it,
or adding exactly one of the children of an OR-gate in $\rho$ to it (and never the other).
More formally,
it can be seen as a connected subcircuit of $D$ that contains $r$,
and for every AND-gate it contains it must either contain none or all of its children,
and for every OR-gate none or exactly one of its children.
(Weights on wires are inherited as well.)
The maximal precertificates of $D$ are exactly its certificates.

For a precertificate $\rho$,
define $\val(\rho)$ to be the value of $\rho$ after replacing all its leaf gates with 1,
or equivalently the product of all weights on the internal wires of $\rho$.
This agrees with the earlier definition of $\val(\rho)$ if $\rho$ is a certificate.

Define $\psi(\rho)$
to be the vector where $\psi(\rho)_{c_g}$ is $\ket 1$ if $g \in \rho$ and $\ket 0$ otherwise,
and $\psi(\rho)_{v_x}$ is $\ket 1$ if $x \in \rho$,
$\ket 0$ if $\n x \in \rho$,
and $\ket +$ otherwise.

Let $\pre_h(D)$ be the set of maximal precertificates
of $D$ restricted to gates $g \preceq h$ and their children.

\begin{lemma}\label{lem:p-layer}
	After applying all $P_g, g \preceq h$,
	$Q(D)$ is in state
	$\sum_{\rho \in \pre_h(D)}\val(\rho)\psi(\rho)$.
\end{lemma}
\begin{proof}
	By induction on $h$.
	It holds for $h = r$.
	(It helps to assume here without loss of generality that $r$ is a unary OR-gate of weight 1.)
	Assume it holds for all $h' \prec h$.
	Let $g$ be the predecessor of $h$ in the topological order,
	and let $\rho \in \pre_g(D)$.

	If $h \notin \rho$
	then $\rho$ is also in $\pre_h(D)$,
	and $P_h$ does not act on $\psi(\rho)$.
	If $h \in \rho$,
	there are three cases:
	\begin{itemize}
		\item $h$ is an AND-gate.
			Then $\rho$ is not in $\pre_h(D)$
			but its unique extension $\rho'$
			with exactly the children of $h$ added to it is an element of $\pre_h(D)$.
			$P_h$ acts on $\psi(\rho)$
			and turns it into $\psi(\rho')$ as desired.
		\item $h$ is a unary OR-gate with weight $w$.
			Then $\rho$ is not in $\pre_h(D)$ but its unique extension $\rho'$
			with exactly the child of $h$ added to it is an element of $\pre_h(D)$.
			$P_h$ acts on $\psi(\rho)$
			and turns it into $w\psi(\rho')$ as desired.
			($\val(\rho') = w\val(\rho)$.)
		\item $h$ is an OR-gate $a_0 \lor a_1$ with weights $w_0,w_1$.
			Then $\rho$ is not in $\pre_h(D)$
			but its two extensions $\rho_0, \rho_1$ are,
			with the single added gates $a_0,a_1$ respectively.
			$P_h$ acts on $\psi(\rho)$
			and turns it into
			$w_0\psi(\rho_0) + w_1\psi(\rho_1)$
			as desired.
			\qedhere
	\end{itemize}
\end{proof}

\begin{corollary}\label{cor:p-layer}
	After the first layer,
	$Q(D)$ is in state
	$\sum_{\tau \in \cert(D)} \val(\tau) \psi(\tau)$.
\end{corollary}

Next we note that the $n_g$ capture the nonzero gates.
Define $\nu(x)$ by
$\nu(x)_{n_g}$ being $\ket 1$
if the gate $g$ is nonzero in the circuit $D(x)$
and $\ket 0$ otherwise.

\begin{lemma}\label{lem:n-layer}
	After the second layer,
	$Q(D)$ is in state 
	\begin{align*}
		\sum_{x} \frac{D(x)}{\sqrt{m_D(x)}} \ket x
		\otimes \psi(T_D(x))_{c_g : g \in \I}
		\otimes \nu(x).
	\end{align*}
\end{lemma}
\begin{proof}
	At the end of the first layer,
	$Q(D)$ is in state
	\begin{align*}
		\sum_{x} \frac{D(x)}{\sqrt{m_D(x)}} \ket x
		\otimes \psi(T_D(x))_{c_g : g \in \I}
		\otimes \ket 0^{\otimes|\I \setminus \{ r \}|}.
	\end{align*}
	That the subcircuits $N_g$ now act as desired
	follows from Boolean arithmetic.
\end{proof}

\begin{lemma}\label{lem:r-layer}
	The third layer resets the $c_g$ and $n_g$ to $\ket 0$,
	and leaves the $v_x$ untouched.
\end{lemma}
\begin{proof}
	The $n_g$ are set to $\ket 0$
	because they only act as control
	and are then reset by the inverse of the second layer at the end of the third.
	The $v_x$ only act as control in this layer,
	so they are left untouched.

	It remains to show that the $c_g$ are reset.
	Since $c_r$ is reset by the final $X$-gate,
	we need only consider internal $g \neq r$.
	For this we look at each summand vector
	$\pi(x) = \ket x \otimes \psi(T_D(x))_{c_g : g \in \I} \otimes \nu(x)$ individually.
	All $R_g$ with $g \notin \tau$ do not act on it,
	since they are controlled on $c_g$.
	We have to show that the $R_g$ with $g \in \tau$ reset all $c_h, h \in \tau \setminus \{ r \}$
	(which together with $c_r$ are precisely the $c_a$ in $\pi(\tau)$ set to $\ket 1$).

	The $R_g$ act only on the $c_a$ where $a$ is a child of $g$.
	Because certificates are trees
	(\autoref{fact:tree-certs}),
	every $a \in \tau \setminus \{ r \}$ has exactly one parent $p_a$ in $\tau$.
	Hence it remains to show that each $R_g, g \in \tau$
	negates $c_a$ for its children $a \in \tau \cap \I$
	and leaves all other qubits untouched.
	For AND-gates $g$,
	we note that if $g$ is in $\tau$,
	all its children $a_i$ are too,
	and that $R_g$ on $\pi(\tau)$
	negates all $c_{a_i}, a_i \in \I$ and leaves other qubits untouched.
	For OR-gates $g \in \tau$,
	exactly one child is in $\tau$.
	Because $D$ is deterministic,
	this is the one child $b$ of $g$ that is nonzero
	in every computation of $D(x)$ where $\tau = T(x)$,
	which is the one child with $\pi(\tau)_{n_b} = \ket 1$.
	It follows that $R_g$ on $\pi(\tau)$
	negates $c_b$ and leaves all other qubits untouched.
\end{proof}

Combining \autoref{cor:p-layer}
and Lemmas~\ref{lem:n-layer},~\ref{lem:r-layer},
we conclude:

\begin{corollary}[Correctness]\label{cor:correctness}
	If $D$ is a normalized weighted 2-d-DNNF and $S_2(D)$ is a quantum state,
	then $Q(D)$ computes $S_2(D)$.
\end{corollary}

\section{Complexity}\label{sec:complexity}
We now look at the dimensions of $Q(D)$
and the complexity of generating it.

\begin{lemma}\label{lem:dimensions-with-CCH}
	$Q(D)$ consists of $O(|D|)$ gates,
	has width $|\var(D)| + 2|\I| - 1 < 2|D|$,
	and depth $O(\depth(D))$.
\end{lemma}
\begin{proof}
	The number of ancillas used is precisely
	$2|\int(D)| - 1$:
	one $c_g$ for every $g \in \int(D)$ and one $n_g$ for every $g \in \int(D) \setminus \{ r \}$.
	This gives the width bound.

	To prove that $\depth(Q(D)) = O(\depth(D))$,
	we note that gates $P_g$ for all $g$ in the same level commute,
	in the sense that their application order does not affect 
	the outcome when applied to
	the state obtained by the circuit for all gates $P_{g'}$ of earlier levels.
	The same observation can be made for $R_g, N_g$.
	This follows from determinism, decomposability,
	and in particular that all certificates are trees.
	(It can also be seen as a corollary
	of the fact that the correctness proofs in \autoref{sec:correctness}
	apply to any topological order.)
	It thus suffices to show that for each level,
	the $P_g, R_g, N_g$ can be scheduled in constant depth.

	Write $P^i$ for the set of $P_g$ with $g$ of level $i$.
	Say that $P_g,P_{g'} \in P^i$ \emph{conflict}
	if they share a qubit and $g,g'$ are distinct.
	Consider the conflict graph $G$ on $P^i$,
	where there is an edge between $P_g,P_{g'}$ if and only if they conflict.
	Because literals are unique and fan-out and fan-in are bounded,
	the degree of $G$ is bounded
	($\deg G \leq 6$)
	and thus admits a finite graph coloring (using at most 7 colors).
	$P_g$ of the same color can be scheduled in parallel,
	hence $P^i$ can be scheduled in constant depth.
	A similar argument can be made for the $N_g$ and $R_g$.
	It follows that $\depth(Q(D)) = O(\depth(D))$ on gate set $\gateset$.
	(More careful counting shows that in terms of $P_g,R_g,N_g,N^{-1}_g$,
	the depth of $Q(D)$ is bounded by $11\depth(D) + O(1)$.)
\end{proof}

\begin{lemma}\label{lem:generating-complexity}
	Up to complex arithmetic,
	$Q(D)$ can be described in $O(|D|)$ time sequentially
	or in parallel constant time using $\poly(|D|)$ processors.
\end{lemma}
\begin{proof}
	The descriptions of subcircuits $P_g,R_g,N_g$
	are of constant size
	and depend only on $g$ and (the parity of) its literal input gates,
	and their inputs and outputs depend only on the inputs and outputs of $g$
	(they are amongst
	qubits $c_g,n_g$,
	$c_{a_i},n_{a_i}$ for inputs and outputs $a_i$ of $g$,
	and $v_x$ for variables $x$ whose literal gates are inputs of $g$).
	Thus, assuming a reasonable encoding of $D$
	where the type of a gate
	(in the case of literals, the parity and variable)
	can be obtained in constant time,
	a description of $Q(D)$ can be obtained in linear time from a description of $D$,
	or in constant parallel time by describing each of the subcircuits in parallel.

	Though we do not consider it part of the circuit description,
	the $O(\depth(D))$ scheduling of \autoref{lem:dimensions-with-CCH}
	can be computed in linear time,
	since the colorings in the proof can be found with the standard linear-time greedy algorithm.
\end{proof}

As a final optimization,
we note that qubit $c_r$ is $\ket 1$ throughout the computation of $Q(D)$ (until it is set to $\ket 0$ by an uncontrolled $X$ in the reset layer),
and can thus easily be eliminated.
Combining our results thus far
together with the gate set reduction in \autoref{app:gate-set}
implies \autoref{thm:s2}.

\section{Preparing \texorpdfstring{$S_1(D)$}{S1(D)}}\label{sec:reweighting}
We have shown how to prepare states $S_2(D)$ for normalized weighted 2-d-DNNF $D$.
But what if we want to prepare $S_1(D)$,
or $S_2(D)$ for some non-normalized weighted 2-d-DNNF?
In this section,
we show that weighted 2-d-DNNF $D$ where $S_1(D)$ is a quantum state
can be efficiently reweighted to obtain a weighted 2-d-DNNF $D'$ with $S_2(D') = c S_1(D)$ for some $c \in \R_{>0}$,
and also that weighted 2-d-DNNF can be efficiently normalized while preserving
the $S_2$-state it describes.
It follows that states $S_1(D),S_2(D)$ can be prepared for any weighted 2-d-DNNF.

A \emph{reweighting} $D'$ of a weighted d-DNNF $D$
is a weighted d-DNNF obtained from $D$ by multiplying its weights by nonzero factors.
Reweightings have the same certificates,
but possibly with different nonzero values.
In particular, $D(x) \neq 0$ if and only if $D'(x) \neq 0$.
Since reweightings have the same gate structure,
below we will often write $g'$ for the gate of $D'$ corresponding with $g$,
and if $w_i$ is a weight of an OR-gate child $c_i$,
we will write $w_i'$ for the weight of $c_i'$ in $D'$.

We first show how to obtain $S_2(D')$ up to a constant factor.
In this work, $a \propto b$ means that there is some real number $c > 0$ with $a = cb$.
\begin{lemma}\label{lem:s1-to-s2}
	For every weighted 2-d-DNNF $D$ with $S_1(D) \neq 0$
	there is a reweighting $D'$
	with $S_2(D') \propto S_1(D)$
	that is computable in linear time up to complex arithmetic.
\end{lemma}
\begin{proof}
	It suffices to find a reweighting $D'$
	so that $D'(x) \propto \sqrt{m_D(x)}D(x)$.
	Equivalently,
	it suffices that $\val_{D'}(\tau) \propto 2^{- |\var(\tau)|/2}\val_D(\tau)$
	for all certificates $\tau$,
	where with $\val_{D'}(\tau)$ we refer to the value of certificate $\tau$ taking edge weights from $D'$.
	Define $K_g$ recursively by
	$K_x = K_{\n x} = 2^{-1/2}$;
	$K_{a \times b} = K_a K_b$;
	and $K_{\sum_i w_i c_i} = 1$.
	Define the reweighting $D'$ by $w_i' = w_i K_{c_i}$.

	Intuitively, this reweighting makes every $w_i'$ compensate
	for newly introduced variables not yet compensated for by lower $+$-gates.
	Formally,
	we prove by induction on $g$ that
	$\val_{D'}(\tau) = 2^{- |\var(\tau)|/2} \val_D(\tau) / K_g$
	for any certificate $\tau \in \cert(g)$:
	\begin{itemize}
		\item Literals $g$ have one certificate $\tau = \{ g \}$
			and $\val_{D'}(\tau) = 1 = 2^{-1/2} \cdot 1 / 2^{-1/2}$.
		\item For certificates $\tau$ of gates $g = a \times b$
			there are unique
			$\sigma \in \cert(a), \rho\in \cert(b)$
			with $\tau = \{g \} \cup (\sigma \sqcup \rho)$
			and we find
			\begin{align*}
				\val_{D'}(\tau)
				&= \val_{D'}(\sigma) \cdot \val_{D'}(\rho)
				\\&= 2^{- |\var(\sigma)|/2} \val_D(\sigma) / K_a
				\cdot 2^{- |\var(\rho)|/2} \val_D(\rho) / K_b
				\\&= 2^{- (|\var(\sigma)| + |\var(\rho)|) / 2}
				\val_D(\sigma)\val_D(\rho) / (K_a K_b)
				\\&= 2^{- |\var(\tau)|/2} \val_D(\tau) / K_g.
			\end{align*}
		\item For $g = \sum_i w_i c_i$ and certificate $\tau \in \cert(g)$
			containing $c_i$
			and writing $\tau'$ for $\tau$ without $g$,
			we have
			\begin{align*}
				\val_{D'}(\tau)
				&= w_i' \val_{D'}(\tau')
				= w_i K_{c_i} \cdot 2^{- |\var(\tau')|/2}\val_{D}(\tau') / K_{c_i}
			      \\&= 2^{- |\var(\tau')|/2} \val_D(\tau)
				= 2^{-|\var(\tau)|/2} \val_D(\tau) / K_g.
			\end{align*}
	\end{itemize}
	Hence,
	writing $r$ for the output gate,
	we have
	$\val_{D'}(\tau) = 2^{-|\var(\tau)|/2} \val_D(\tau) / K_r$
	for all certificates $\tau$,
	and thus $D'(x) \propto \sqrt{m_D(x)} D(x)$
	as desired.
	By the recursive definition of $K_g$,
	$D'$ can be computed in linear time up to complex arithmetic.
\end{proof}
It remains to normalize.

\begin{lemma}\label{lem:normalization}
	For every weighted 2-d-DNNF $D$ with $S_2(D) \neq 0$
	there is a normalized reweighting $D'$ with $S_2(D') = S_2(D) / |S_2(D)|$.
	It can be computed in linear time up to complex arithmetic.
\end{lemma}
\begin{proof}
	For gates $g = \sum_i w_i c_i$
	we have $S_2(g) = \sum_i w_i \tilde S_2(c_i)$,
	where we write $\tilde S_2(c_i)$
	for the appropriate extension $S_2(c_i) \otimes \ket +^{\otimes |\var(g)\setminus\var(c_i)|}$
	of $S_2(c_i)$ to $\var(g)$.
	Note that $|\tilde S_2(c_i)| = |S_2(c_i)|$.
	Define $D'$ by $w_i' = w_i |\tilde S_2(c_i)| / |S_2(g)|$.
	It can be shown by induction
	that $g'(x) = g(x) / |S_2(g)|$ for all gates $g$,
	hence $D'(x) = D(x) / |S_2(D)|$
	and $S_2(D') = S_2(D) / |S_2(D)|$.
	Furthermore, $D'$ is normalized:
	\begin{align*}
		\sum_i |w_i'|^2
		&= \sum_i |w_i|^2 |S_2(c_i)|^2 / |S_2(g)|^2
		= \sum_i |w_i|^2 |\tilde S_2(c_i)|^2 / |S_2(g)|^2
	      \\&= \left(\sum_i |w_i \tilde S_2(c_i)|^2\right) / |S_2(g)|^2
		= |S_2(g)|^2 / |S_2(g)|^2
		= 1.
	\end{align*}
	It remains to show that
	all $|S_2(g)|$ together can be computed in linear time up to complex arithmetic.
	This can be done using the recurrence
	\begin{align*}
		|S_2(x)| &= |S_2(\n x)| = 1;\\
		|S_2(a \times b)| &= |S_2(a)| |S_2(b)|;\\
		|S_2(\textstyle\sum_i w_i c_i)| &= \sqrt{\textstyle\sum_i |w_i|^2 |S_2(c_i)|^2}.
		\qedhere
	\end{align*}
\end{proof}

Combining lemmas~\ref{lem:s1-to-s2}~and~\ref{lem:normalization}, we in particular obtain:

\begin{corollary}\label{lem:reweighting}
	There is an algorithm reweighting weighted 2-d-DNNF $D$
	with $S_1(D)$ a state
	into normalized weighted 2-d-DNNF $D'$
	with $S_1(D) = S_2(D')$.
	It takes linear time up to complex arithmetic.
\end{corollary}

In practice, it may be more efficient to combine the two procedures.
In the proof of \autoref{lem:s1-to-s2},
changing the definition of $K_g$ for $+$-gates to
$K_{\sum_i w_i c_i} = \sqrt{\sum_i|K_{c_i}w_i|^2}$
and $w_i'$ to $w_i K_{c_i} / K_g$
makes $D'$ a normalized reweighting with $S_2(D') = S_1(D)$.

\section{Related work}\label{sec:related-work}
\subsection{Weighted FBDD}
Recently, Tanaka, Yamasaki, and Murao \cite{tanaka2026qsp}
showed that a state represented by a weighted free binary decision diagram (weighted FBDD)
$F$ with $N$ nodes can be computed by a quantum circuit
of size and depth $O(|F|)$ using $N$ ancillas.
In \cite[Table 1]{tanaka2026qsp},
they list various other state preparation methods
and show that their approach compares favorably
if the to-be-prepared state can be described as a succinct (polynomial in the number of qubits) weighted FBDD.

Weighted FBDD is a sublanguage of weighted d-DNNF \cite{darwiche_knowledge_2002}:
there is a straightforward linear-time conversion from weighted FBDD $F$ to weighted d-DNNF.
\autoref{cor:main} then gives a quantum circuit of size and width $O(|F|)$
and depth $O(\min(|F|, \depth(F)\log |F|))$,
so our approach asymptotically subsumes that of \cite{tanaka2026qsp}.
It does so strictly,
since weighted d-DNNF are exponentially more succinct than weighted FBDD
(\cite{beame2013}, \autoref{app:succinctness}).
That is, there are families of quantum states
that can be represented by weighted d-DNNF of polynomial size (in the number of qubits)
but only by exponential-size weighted FBDD.\footnote{
	In particular,
	deterministic tree tensor networks can be transformed into weighted d-DNNF in linear time,
	while there are families of states with small deterministic tree tensor networks
	but no small weighted FBDD
	and vice versa
	(\cite{bollig_relative_2019,savicky_read-once_2000,capelli2016structural,quist2026}, \autoref{app:succinctness}).
Thus our approach is significantly more general.}

The linear-time conversion from weighted FBDD to weighted d-DNNF
does increase the size by a constant factor,
so in practice using \cite{tanaka2026qsp} directly may be more ancilla-efficient
if the state is represented by a small weighted FBDD and there is no much smaller weighted d-DNNF.

The method of Tanaka et al. works roughly by
creating a superposition over the partial paths in the weighted FBDD,
using an ancilla for every node to keep track of which nodes are visited in each path,
and increasing the path lengths until a superposition over all paths is obtained.
Each basis state has a unique path in the weighted FBDD,
and the amplitude of $\ket x$ in the represented state is the amplitude of the path,
so the preparation is completed by uncomputing the ancillas.
Our preparation algorithm,
on the other hand,
makes use of the fact that each basis state $\ket x$
with $D(x) \neq 0$ has a unique certificate.
We create a superposition over all partial certificates,
increasing depth until a superposition over all certificates is reached.
In contrast with weighted FBDD,
uncomputing the ancilla wires posed a significant new challenge,
which we solved
by computing (in superposition) for each basis state which gates of $D$
are nonzero into the additional $n_g$-wires
and then using this to uncompute the certificate ancillas
(cf. layers 2, 3 of \autoref{sec:impl}).

\subsection{Other related work}
Because, asymptotically speaking,
our work subsumes that of Tanaka et al.,
their extensive overview and comparison with other state preparation methods
\cite[Table 1]{tanaka2026qsp}
applies to our method as well.
Below,
we mainly list some new comparisons
as well as a stricter comparison with the Grover--Rudolph / Kaye--Mosca approximate QSP approach \cite{grover_creating_2002,kaye_quantum_2004},
where we note that their method can be implemented efficiently from a weighted d-DNNF.
In general, since QSP methods make different assumptions on the to-be-prepared quantum state
it is tricky to make good comparisons without analyzing the relative succinctness of the input languages.
Many methods also focus on approximate state preparation,
where instead of the actual to-be-prepared state only an approximation is prepared (within a given error bound),
while ours and that of Tanaka et al. are exact.

\subsubsection{Knowledge-compilation-based state preparation}
Apart from \cite{tanaka2026qsp},
knowledge-compilation-based state preparation work
for weighted OBDD \cite{mozafari_efficient_2022}
(a subclass of weighted FBDD)
and extensions for LIMDD \cite{hong_limtdd_2026}
do not guarantee low complexity in the size of the representation,
but in the number of paths,
which in the worst case is exponential in the size of the decision diagram.

Compared to other well-known knowledge representations,
d-DNNF is relatively succinct.
Other well-known succinct representations are non-deterministic,
notably SDNNF
(where there is no determinism but extra structural requirements on variable decomposition)
and non-deterministic FBDD
(where decision nodes can have multiple edges per decision).
d-DNNF, SDNNF, and non-deterministic FBDD are pairwise incomparable,
in the sense that for each two there are families of Boolean functions polynomially represented in the one and not in the other \cite{amarilli_connecting_2020}.

In \cite{quist2026} it is shown that weighted SDNNF are equivalent with tree tensor networks (TTN),
and weighted non-deterministic FBDD with matrix product states (MPS).
For both MPS and TTN,
to the best of our knowledge,
current tractable approaches are only approximate or assume low bond dimension \cite{malz_preparation_2024,gonzalez-conde_efficient_2024,rudolph_mps_2024,manabe_state_2025}.

\subsubsection{Grover--Rudolph preparation}
While we concern ourselves with exact preparation,
we note that the general approximate state preparation method of \cite{grover_creating_2002,kaye_quantum_2004}
can be efficiently implemented if the state is given as a weighted d-DNNF.

In this method,
the amplitudes of the to-be-prepared state are first computed
into quantum registers.
For a given weighted d-DNNF $D$,
it roughly suffices to construct two gadgets:
a subcircuit
that on input $\ket{b_0 \dots b_i}$ and a zeroed $m$-bit quantum register
computes into the register an $m$-bit approximation of
$\sum_{x : \land_{j \leq i} x_j = b_j} |D(x)|^2$,
and a subcircuit that given $\ket x$ and a zeroed $m$-bit quantum register
computes an $m$-bit approximation of the phase of $D(x)$ into the register.
Both can be implemented by evaluating modifications of the arithmetic circuit underlying $D$.
By \cite{miller_efficient_1986},
these evaluations can be done by circuits of
depth $\polylog(|D|)$ and size $\poly(|D|)$.
This leads to a quantum circuit of size $\poly(m, |D|)$
and depth $O(|\var(D)|\polylog(m,|D|))$
that
approximates the state represented by $D$
(for a good choice of $m$ based on $D$ and the desired precision).
In comparison,
our exact approach has sharper size and width (ancilla) bounds as well as a matching depth bound.
(This uses that weighted d-DNNF $D$ can be made depth $O(|\var(D)|)$
by merging consecutive OR-nodes.)

\bibliographystyle{plainurl}
\bibliography{references.bib}

\appendix

\section{Reducing the gate set}\label{app:gate-set}
We briefly show that our results still hold
when we restrict our gate set to $X,CX$ and the $U(p,q)$.
\begin{lemma}\label{lem:gate-set}
	Gate set $\{ X, CX \} \cup SU(2)$ suffices.
\end{lemma}
\begin{proof}
	We have to show that the other gates in $\gateset$
	(defined in \autoref{subsec:qsp})
	used in $Q(D)$
	can be replaced appropriately.

	By \cite[Lemma 5.1]{barenco_elementary_1995},
	the $CU(p,q)$ can be implemented with two $CX$
	and three $SU(2)$-gates.
	$H$-gates can be implemented using
	$H = X \circ U(2^{-1/2},2^{-1/2})$,
	and likewise
	$CH = CX \circ CU(2^{-1/2},2^{-1/2})$.

	Instead of implementing $CCX$ and $CCH$ exactly,
	we replace their occurrences in $Q(D)$
	by subcircuits that act appropriately on the basis states that actually occur.
	We use \cite[Lemma 6.1]{barenco_elementary_1995}:
	for any single-qubit gate $A$,
	the doubly-controlled $CCA$ can be implemented
	with a constant number of $CX$ and $CV, CV^\dagger$,
	where $V^2 = A$.
	We will choose appropriate $CU(p,q)$ for $CV$
	to replace instances of $CCX$ and $CCH$ in $Q(D)$.
	(Note that $CU(p,q)^\dagger = CU(p^*, -q)$.)

	In the first two layers,
	the $CCX$ only fire on $\ket 0$,
	so it suffices to find a $CV$
	with $V^2\ket 0 = \ket 1$.
	This is satisfied by $CU(2^{-1/2},2^{-1/2})$.
	In the reset layer,
	the $CCX$ only fire on $\ket 1$
	and we can choose $CU(2^{-1/2},-2^{-1/2})$ for $CV$.
	The $CCH$ gates only fire on $\ket +$
	and we can take $CU(\cos(\pi/8), -\sin(\pi/8))$ for $CV$.
\end{proof}

\section{Succinctness of weighted knowledge compilation languages for quantum state representation}\label{app:succinctness}
While succinctness relations for unweighted knowledge representation languages are well-known
(see in particular \cite{darwiche_knowledge_2002} and \url{https://circuitzoo.net}),
the relative succinctness of their weighted variants
(for representing quantum states)
is less fleshed out,
and proofs do not necessarily follow from the unweighted results.

As noted in \cite{colnet_compilation_2021},
succinctness results for deterministic unweighted knowledge representation languages
often do carry over to their weighted variants.
In particular, weighted d-DNNF are exponentially more succinct than weighted FBDD.
Briefly:
let $F$ be a family of Boolean functions that can be represented
by polynomial-size d-DNNF but only by exponential-size FBDD \cite{beame2013}.
Consider the family of quantum states $\Psi = \{ \psi_f : f \in F \}$
with $\psi_f = \sum_x N^{-1/2} f(x) \ket x$,
where $N$ is the model count of $f$.
A small weighted d-DNNF representing $\psi_f$ is easily obtained from a small d-DNNF representing $f$.
Assume towards a contradiction that the states in $\Psi$ can be represented by small weighted FBDD $D_f$.
Modify $D_f$ so that all edges of weight $0$ point directly to the $0$-sink,
and then remove all edge weights to obtain a succinct FBDD representing $f$.
This contradicts the assumption that $F$ cannot be represented succinctly in FBDD.
It follows that weighted FBDD is exponentially less succinct than weighted d-DNNF
for representing quantum states.

Similar arguments can be made for all the separations of deterministic representations
cited in this work.
For non-deterministic representations,
the addition of weights can lead to cancellation
and the argument above generally does not apply.

\section{Weighted \texorpdfstring{$\textup d$-DNNF}{d-DNNF} with constants}\label{app:constant-free}
\begin{lemma}\label{lem:constant-free}
	A nonconstant weighted d-DNNF $D$ with constant gates
	can be made constant-free in linear time
	up to complex arithmetic.
	The result $D'$ satisfies
	$|D'| \leq 2|D|$
	and $\depth(D') \leq 2\depth(D)$.
\end{lemma}
\begin{proof}
	Remove all gates that compute the zero function.
	Identifying these gates can be done altogether in linear time
	using the algorithm of \cite[Section 2.2]{darwiche_decomposable_2001}.
	Next, using depth-first traversal,
	find the set of gates $g$ that satisfy $\var(g) = \varnothing$.
	These are nonzero (we have removed zero gates),
	and are replaced with their constant value by evaluation,
	which takes linear time up to complex arithmetic.
	Now all $+$-gates take some non-constant gate as input,
	and so by determinism they do not take on any constant gate as input.
	The remaining constants are removed
	by replacing $\times$-gates
	with non-constant children $a_i$
	and constant children $c_j$
	by $\times$-gates with children $a_i$
	feeding into a unary $+$-gate with weight $\prod_j c_j$.
\end{proof}

\end{document}